\documentclass[3p,times,twocolumn,sort&compress]{elsarticle}

\usepackage{amssymb}
\usepackage{amsmath}

\usepackage{amsmath,amssymb,amsfonts}
\usepackage{algorithmic,algorithm}
\usepackage{graphicx}
\usepackage{textcomp}

\usepackage{amsthm,amssymb}
\usepackage{mathrsfs}
\usepackage{enumerate}
\usepackage{float}
\usepackage{color}
\usepackage{bm}
\usepackage{subcaption}
\usepackage{caption}
\usepackage{tikz}
\usetikzlibrary{arrows.meta, decorations.pathmorphing}
\usetikzlibrary{decorations.pathmorphing, patterns, arrows.meta}
\usetikzlibrary{decorations.pathmorphing, patterns, calc}
\usepackage{amsmath}

\usepackage{hyperref}
\usepackage{algorithm}
\usepackage{footnote}
\usepackage{multirow}
\usepackage{setspace}
\usepackage{balance}

\makeatletter

\newcommand{\Rmnum}[1]{\expandafter\@slowromancap\romannumeral #1@}

\newtheorem{assumption}{Assumption}
\newtheorem{lemma}{Lemma}
\newtheorem{remark}{Remark}
\newtheorem{theorem}{Theorem}
\newtheorem{definition}{Definition}

\newtheorem{corollary}{Corollary}

\newcounter{ltisystem}
\newcommand{\spr}{{\rm spr }}

\journal{Journal of the Franklin Institute}

\begin{document}

\begin{frontmatter}



\title{Finite Sample Analysis of System Poles for Ho-Kalman Algorithm}


\author[1]{Shuai Sun\corref{cor1}%
}
\ead{suns19@mails.tsinghua.edu.cn}
\author[2]{Xu Wang}
\ead{worm@mail.ustc.edu.cn}
\cortext[cor1]{Corresponding author}

\affiliation[1]{organization={Department of Automation, Tsinghua University},
city={Beijing},
postcode={100084},
country={China}}
\affiliation[2]{organization={School of Computer Science and Technology, University of Science and Technology of China},
city={Hefei,Anhui},
postcode={230026},
country={China}}

\begin{abstract}
The Ho-Kalman algorithm has been widely employed for the identification of discrete-time linear time-invariant (LTI) systems. In this paper, we investigate the pole estimation error for the Ho-Kalman algorithm based on finite input/output sample data. Building upon prior works, we derive finite sample error bounds for system pole estimation in both single-trajectory and multiple-trajectory scenarios. Specifically, we prove that, with high probability, the estimation error for an $n$-dimensional system decreases at a rate of at least $\mathcal{O}(T^{-1/2n})$ in the single-trajectory setting with trajectory length $T$, and at a rate of at least $\mathcal{O}(N^{-1/2n})$ in the multiple-trajectory setting with $N$ independent trajectories. Furthermore, we reveal that in both settings, achieving a constant estimation error requires a super-polynomial sample size in $ \max\{n/m, n/p\} $, where $n/m$ and $n/p$ denote the state-to-output and state-to-input dimension ratios, respectively. Finally, numerical experiments are conducted to validate the non-asymptotic results of system pole estimation.
\end{abstract}




\end{frontmatter}

\section{Introduction}
\label{sec:introduction}

Linear time-invariant (LTI) system theory is widely used in finance, biology, robotics, and other engineering fields~\cite{ljung1986system}. 
Many modern control techniques depend on the availability of a reasonably accurate LTI state-space model of the system to be controlled~\cite{oymak2021revisiting}. 
To this end, identifying the state-space model of the following multiple-input multiple-output (MIMO) LTI system has become one of the core problems in system analysis, design and control~\cite{michel1992}:
\begin{equation}\label{linear_system}
	\begin{aligned}
		x_{k+1} &= A x_{k}+ Bu_k + 
		w_{k}, \\
		y_{k} &= 
		C x_{k}+ Du_k +  v_{k},
    \end{aligned}
    \tag{\textsf{LTI-system}}
\end{equation}
where $x_k \in \mathbb{R}^n$, $u_k \in \mathbb{R}^p$, $y_k \in \mathbb{R}^m$ are the system state, the input and the output, respectively; $w_k \in \mathbb{R}^n$ and $v_k \in \mathbb{R}^m$ denote the process and measurement noise.

Numerous of classical works on the identification of LTI systems~\cite{kodde1990asymptotic,gibson2000least,deistler1995consistency,viberg1997analysis} mainly focused on the asymptotic properties
of certain estimation methods~\cite{goodwin1977dynamic,qin2006overview}, such as consistency~\cite{jansson1998consistency}, asymptotic normality~\cite{west1988asymptotic}, and variance analysis~\cite{jansson2000asymptotic}.
However, in real-world applications, only a finite amount of sample data is available. Consequently, obtaining sharp error bounds under finite-sample conditions is crucial for designing high-performance robust control systems. Motivated by this, recent research in both academia and industry has shifted its focus toward non-asymptotic analysis, which studies system identification under finite-sample settings using statistical tools~\cite{matni2019tutorial,matni2019self,dean2020sample}.
Especially, driven by advances in statistical learning and high-dimensional probability~\cite{vershynin2018high,wainwright2019high}, an extensive body of work~\cite{faradonbeh2018finite,simchowitz2018learning,simchowitz2019learning,matni2019tutorial,matni2019self,sarkar2019near,tsiamis2019finite,tsiamis2021linear,tsiamis2022learning,tsiamis2022online,jedra2019sample,jedra2020finite,jedra2022finite,dean2020sample,sun2020finite,zheng2020non,oymak2021revisiting,wang2022large,he2024weighted,chatzikiriakos2024sample,wagenmaker2020active,sun2025non} has emerged to address the parameter identification problem for LTI systems under finite-sample conditions.
Moreover, a comprehensive survey of developments in non-asymptotic system identification can be found in~\cite{tsiamis2023statistical}.

The system described in \eqref{linear_system} can be generally categorized as either \emph{fully observed} or \emph{partially observed} system, depending on the availability of direct state measurements.
When the system state can be measured accurately, i.e., \( C = \mathbf{I}, D = \mathbf{0}\) and $v_k \equiv \mathbf{0}$, the \eqref{linear_system} is known as a \textit{fully observed} LTI system. A series of recent studies~\cite{faradonbeh2018finite,simchowitz2018learning,sarkar2019near,matni2019tutorial,wagenmaker2020active,dean2020sample,jedra2020finite} has established finite-sample upper bounds on the identification error of the system matrices \((A, B)\) under ordinary least-squares (OLS) estimation. 
Beyond OLS, Jedra and Proutiere~\cite{jedra2019sample,jedra2022finite} derived lower bounds on the sample complexity in the probably approximately correct (PAC) framework~\cite{valiant2013probably} for a class of locally stable algorithms. Tsiamis and Pappas~\cite{tsiamis2021linear} further showed that in under-actuated and under-excited systems with certain structural constraints, the worst-case sample complexity grows exponentially with the system dimension. Chatzikiriakos and Iannelli~\cite{chatzikiriakos2024sample} established both upper and lower bounds on the sample complexity for identifying the system matrices \((A, B)\) from a finite set of candidates, even without stability assumptions.

When  \( C \neq \mathbf{I} \), the \eqref{linear_system} is a  \textit{partially observed} LTI system, where identification becomes more challenging due to the inability to accurately measure the system state~\cite{zheng2021sample}. Recent works~\cite{oymak2021revisiting,simchowitz2019learning,zheng2020non,wang2022large,he2024weighted,tsiamis2019finite,sun2025non} established high-probability bounds and convergence rates for identifying the state-space realization \( (A, B, C, D) \) under finite-sample conditions (up to similarity), achieving rates of \( \mathcal{O}(1/\sqrt{T}) \) or \( \mathcal{O}(1/\sqrt{N}) \) (up to logarithmic factors), where \( T \) denotes the length of a single trajectory and \( N \) denotes the number of independent sample trajectories. Among these, works~\cite{tsiamis2019finite,sun2025non} employed subspace identification methods (SIM), which extract relevant system subspaces via projections and regressions, followed by realization reconstruction. In contrast, another line of work~\cite{oymak2021revisiting,simchowitz2019learning,zheng2020non,wang2022large,he2024weighted} followed a two-step procedure:
\textbf{\romannumeral1}) estimating Markov parameters using OLS;
\textbf{\romannumeral2}) recovering a balanced realization via the Ho-Kalman algorithm.
Other approaches based on reinforcement learning~\cite{lale2020logarithmic} and gradient descent~\cite{hardt2018gradient} have also been proposed for finite-sample state-space identification.

Moreover, system poles also play a central role in control design, including pole placement~\cite{brasch1970pole}, $\mathcal{H}_\infty$ loop shaping~\cite{papageorgiou1999h}, and the root locus method~\cite{evans1950control}. They critically govern the stability and dynamic response of the closed-loop system. 
However, little attention has been paid to the estimation error of \emph{system poles}, especially for Ho-Kalman algorithm. This paper fills this gap by providing, to the best of our knowledge, the first finite-sample analysis of system pole estimation via the Ho-Kalman algorithm.

Our main contributions are summarized as follows: 
\begin{itemize}
	\item For the \textbf{single-trajectory} setting, building upon the results of Oymak and Ozay~\cite{oymak2021revisiting}, we derive finite-sample bounds for pole estimation. Specifically, we show that the error decays at a rate of at least $\mathcal{O}(T^{-1/2n})$ with high probability, where $n$ is the system order and $T$ the trajectory length.
    
    \item For the \textbf{multiple-trajectory} setting, we derive similar error bounds for pole estimation by extending the results of Zheng and Li~\cite{zheng2020non}. Specifically, we show that the error decays at least as rate $\mathcal{O}(N^{-1/2n})$ with high probability, where $N$ is the number of independent trajectories.

    \item For both settings, we show that achieving a constant estimation error of the system poles requires a \textbf{super-polynomial} sample size in \( \max\{n/m, n/p\} \), where \( m \) and \( p \) are the output and input dimensions, respectively.
\end{itemize}

This paper is organized as follows.
Section~\ref{sec:problem setup} formulates the identification problem and reviews OLS and Ho-Kalman algorithm. Section~\ref{sec:main results} analyzes the finite sample identification error of the system poles for the Ho-Kalman algorithm.
Finally, Section~\ref{sec:numerical results} provides numerical simulations.

\textbf{Notations}: 
$\mathbf{0}$ is an all-zero matrix. For any $ x \in \mathbb{R} $, $ \lfloor x \rfloor $ denotes the largest integer not exceeding $ x $, and $ \lceil x \rceil $ denotes the smallest integer not less than $ x $.
The Frobenius norm is denoted by $ \|A\|_{\rm F} $, and $ \|A\| $ is the spectral norm of $A$, i.e., its largest singular value $ \sigma_{\max}(A) $. 
Its $ j $-th largest singular value is denoted by $ \sigma_{j}(A) $, and its smallest non-zero singular value  is denoted by $ \sigma_{\min}(A) $.
The spectrum of a square matrix $ A $ is denoted by $ \lambda(A) $.
The spectral radius of a matrix $A$, denoted by $\spr(A)$, is defined as the maximum absolute value of its eigenvalues.
 The Moore-Penrose inverse of matrix $ A $ is denoted by $ A^\dagger $. 
The $ n \times n $ identity matrix is denoted by $\mathbf{I}_n$, and the identity matrix of appropriate dimension is denoted by $\mathbf{I}$.
The matrix inequality $ A \succ B$ implies that matrix $ A -B $ is positive-definite. 
Multivariate Gaussian distribution with mean $ \mu $ and covariance $ \Sigma $ is denoted by $ \mathcal{N}(\mu,\Sigma) $. 
The big-$\mathcal{O}$ notation $\mathcal{O}(f(n))$ represents a function that grows at most as fast as $f(n)$.
The big-$\Omega$ notation $\Omega(f(n))$ represents a function that grows at least as fast as $f(n)$.
The big-$\Theta$ notation $\Theta(f(n))$ represents a function that grows as fast as $f(n)$.

For readability, the proofs of all results are provided in the Appendix.

\section{Problem Setup}
\label{sec:problem setup}
We consider the identification problem for the \eqref{linear_system} based on finite input/output sample data.
The notations  $A \in \mathbb{R}^{n\times n}$, $B \in \mathbb{R}^{n\times p}$, $C \in \mathbb{R}^{m\times n}$, and $D \in \mathbb{R}^{m\times p}$ are \textbf{unknown} matrices. 

\begin{assumption}\label{observable and controllable}
	For the \eqref{linear_system}, the pair $(A,C)$ is observable, and the pair $(A, B)$ is controllable.
\end{assumption}

\begin{remark}
	Indeed,  the \eqref{linear_system} is minimal in the sense of Assumption~\ref{observable and controllable}, i.e, the state-space realization $(A,B,C,D)$ of the \eqref{linear_system} has the smallest dimension among all state-space realizations with the same input-output relationship. On the other hand, it is worth noting that only the observable part of the system can be identified, and the controllability assumption of the system can ensure that all modes can be excited by the external input $u_k$. Therefore, Assumption~\ref{observable and controllable} is necessary and well-defined.
\end{remark}

\begin{assumption}\label{system order is known}
	\begin{enumerate}
	    \item The order $n$ of the \eqref{linear_system} is known.
        \item For each trajectory, the initial condition of the state $x_0 = \mathbf{0}$ is known.
        \item The process noise $w_k \sim \mathcal{N}(\mathbf{0}, \sigma_w^2 \mathbf{I}_n)$ and the measurement noise $v_k \sim \mathcal{N}(\mathbf{0}, \sigma_v^2 \mathbf{I}_m)$ are independent and identically distributed (i.i.d.) across time $k$, with $\sigma_w, \sigma_v > 0$. 
        \item The input is Gaussian, i.e., $u_k \sim  \mathcal{N}(\mathbf{0},\sigma_u^2 \mathbf{I}_p)$ with $\sigma_u > 0$. The process noise, measurement noise, and input are mutually independent.
	\end{enumerate}
\end{assumption}

\begin{remark}
	It is noteworthy that the input/output sample data correspond to an infinite number of equivalent state-space realizations, which are related through similarity transformations. 
    Consequently, this paper focuses on pole estimation, as the poles are invariant under similarity transformations. Their importance in control design has been discussed in the introduction.
\end{remark}

The goal of this paper is to provide a \textbf{finite-sample analysis of system pole estimation}, namely the eigenvalues of matrix $A$ with the Ho-Kalman algorithm.

\subsection{Least-squares procedure}

The least-squares procedure is used to recover Markov parameters from finite input/output sample data, which subsequently enables the recovery of a balanced state-space realization via the Ho-Kalman algorithm. This can be done in two settings:
\begin{itemize}
    \item \textbf{Single-trajectory}: using $\{(u_k, y_k) \mid 0 \leq k \leq T-1\}$, where $T$ denotes the trajectory length.
    \item \textbf{Multiple-trajectory}: using $\{(u_k^{(i)}, y_k^{(i)}) \mid 0 \leq k \leq K-1,\ 1 \leq i \leq N\}$, where $K$ and $N$ denote the trajectory length and number of trajectories, respectively.
\end{itemize}
In this subsection, we briefly review these two least-squares approaches for estimating Markov parameters.

We first define the matrix $G$, consisting of the first $K$ Markov parameters:
\begin{equation}
    G \triangleq \begin{bmatrix}
		D & CB & CAB &\cdots & CA^{K-2}B
	\end{bmatrix} \in \mathbb{R}^{m \times Kp},
\end{equation}
which can be estimated via the least-squares procedure based on finite input/output sample data. The resulting estimate $\widehat{G}$ is then used to construct a balanced state-space realization using the Ho-Kalman algorithm.

Now, we define the matrix \( F \) as 
\begin{equation}\label{F}
	F \triangleq  \begin{bmatrix}
		\mathbf{0}_{m \times n} & C & CA  & \cdots & CA^{K-2}
	\end{bmatrix} \in \mathbb{R}^{m \times Kn},
\end{equation}
which will be used in the analysis presented in Section~\ref{sec:main results}.

\subsubsection{Single-trajectory setting}
The least-squares procedure introduced next is from~\cite{oymak2021revisiting}. Given the single-trajectory $\{(u_k,y_k)\mid 0 \leq k \leq T-1\}$, we generate $\overline{T}$ subsequences of length $K$, where $T = K + \overline{T} -1$ and $\overline{T} \geq 1$.

Define the stacked input and process noise vectors as
\begin{equation}
    \overline{u}_k \triangleq  \begin{bmatrix}
		u_{k} \\ u_{k-1} \\ \vdots \\ u_{k-K+1}
	\end{bmatrix} \in \mathbb{R}^{Kp}, \quad
	\overline{w}_k \triangleq  \begin{bmatrix}
		w_{k} \\ w_{k-1} \\ \vdots \\ w_{k-K+1}
	\end{bmatrix} \in \mathbb{R}^{Kn}.
\end{equation}
Let the output matrix \( Y \) and input matrix \( U \) be defined as
\begin{align}
    Y &\triangleq \begin{bmatrix}
		y_{K-1} & y_{K} & \cdots & y_{T-1}
	\end{bmatrix}^\top \in \mathbb{R}^{\overline{T} \times m}, \\
    U &\triangleq \begin{bmatrix}
		\overline{u}_{K-1} & \overline{u}_{K} & \cdots &\overline{u}_{T-1}
	\end{bmatrix}^\top \in \mathbb{R}^{\overline{T} \times Kp}.
\end{align}
The estimation of $G$ can be derived from solving the least-squares problem
\begin{equation}
    \widehat{G} = \arg \min_{X \in \mathbb{R}^{m \times K p}} \| Y - U X^\top \|_{\rm F}^2,
\end{equation}
with closed-form solution \( \widehat{G} = (U^\dagger Y)^\top \), where \( U^\dagger = (U^\top U)^{-1} U^\top \) is the left pseudo-inverse of the matrix \( U \).

\subsubsection{Multiple-trajectory setting}
The least-squares procedure introduced next is from~\cite{zheng2020non}. For each sample trajectory $i$, $i = 1,2,\cdots,N$, define the input/output data as
\begin{equation}
    y^{(i)} \triangleq \begin{bmatrix}
		y_0^{(i)}   & \cdots & y_{K-1}^{(i)} 
	\end{bmatrix} \in \mathbb{R}^{m \times K},
\end{equation}
and the structured input matrix
\begin{equation}
    U^{(i)} \triangleq \begin{bmatrix}
		u_0^{(i)}   & u_1^{(i)} & \cdots & u_{K-1}^{(i)}  \\
		& u_0^{(i)} & \cdots & u_{K-2}^{(i)}  \\
		& & \ddots & \vdots \\
		& & & u_0^{(i)} 
	\end{bmatrix} \in \mathbb{R}^{pK \times K}.
\end{equation}
Define the output matrix \( Y \) and input matrix \( U \) as 
\begin{align}
    Y &\triangleq \begin{bmatrix}
		y^{(1)} & \cdots & y^{(N)}
	\end{bmatrix} \in \mathbb{R}^{m \times NK}, \\ 
    U &\triangleq \begin{bmatrix}
		U^{(1)} & \cdots & U^{(N)}
	\end{bmatrix} \in \mathbb{R}^{pK \times NK}.
\end{align}
The estimation of $G$ can be derived from solving the least-squares problem
\begin{equation}\label{G_multi}
	\widehat{G} = \arg \min_{X \in \mathbb{R}^{m \times K p}} \| Y - XU \|_{\rm F}^2,
\end{equation}
with closed-form solution \( \widehat{G} = YU^\dagger  \), where \( U^\dagger = U^\top(UU^\top)^{-1}  \) is the right pseudo-inverse of the matrix \( U \).

\subsection{Ho-Kalman algorithm}
\label{subsec:ho-kalman}
In this subsection, we introduce the Ho-Kalman algorithm~\cite{ho1966effective,oymak2021revisiting} for constructing a state-space realization $(A,B,C,D)$ from the Markov parameter matrix $G$ or its estimate $\widehat{G}$.

Before continuing on, we first define the (extended) observability, controllability and the Hankel matrices $O \in \mathbb{R}^{K_1m \times n}$, $Q  \in \mathbb{R}^{n \times (K_2+1)p}$ and $H  \in \mathbb{R}^{K_1m \times (K_2+1)p}$ of the \eqref{linear_system} as
\begin{align} \label{O}
	O &\triangleq \begin{bmatrix}
		C^\top & (CA)^\top & \cdots & (CA^{K_1-1})^\top
	\end{bmatrix}^\top, \\  \label{Q}
	Q &\triangleq \begin{bmatrix}
		B & AB& \cdots & A^{K_2}B
	\end{bmatrix},\\ \label{H}
	H &\triangleq \begin{bmatrix}
		CB & CAB &\cdots & CA^{K_2}B\\
		CAB & CA^2B &\cdots & CA^{K_2+1}B\\
		\vdots &  \vdots &\ddots &\vdots \\
		CA^{K_1-1}B & CA^{K_1}B &\cdots  &CA^{K_1+K_2-1}B
	\end{bmatrix} = OQ,
\end{align}
where $K = K_1 + K_2 + 1$, and $K_1, K_2 \geq n$. 

We follow the Ho-Kalman algorithm applied to the estimated Markov parameter matrix $\widehat{G}$ as described in~\cite{oymak2021revisiting}, and summarize its main steps below.

\textbf{Step 1}: Construct the estimated Hankel matrix $\widehat{H}$ from $\widehat{G}$. Next, $\widehat{H}$ denotes the estimated value of the Hankel matrix $H$, and similarly for the others. Let $ \widehat{H}^+ $ and $ \widehat{H}^- $ be the sub-matrices of $ \widehat{H} $ discarding the left-most and right-most $ mK_1 \times p $ block respectively. Denote $\widehat{L} \in \mathbb{R}^{K_1m \times K_2p}$ as the best rank-$n$ approximation of $\widehat{H}^{-}$.

\textbf{Step 2}: Perform the singular value decomposition (SVD):
\begin{equation}
    \widehat{L} = \mathcal{U}_1 \Sigma_1 \mathcal{V}^\top_1,
\end{equation}
where $\Sigma_1$ contains the top $n$ singular values in descending order, and $\mathcal{U}_1$ and $\mathcal{V}_1$ are matrices composed of corresponding left and right singular vectors, respectively. 

\textbf{Step 3}: Estimate the (extended) observability and controllability matrices:
\begin{equation}
    \widehat{O} = \mathcal{U}_1 \Sigma_1^{1/2}, \quad \widehat{Q} = \Sigma_1^{1/2}\mathcal{V}_1^\top.
\end{equation}

\textbf{Step 4}: Recover the estimated system matrices:
\begin{equation}
    \begin{aligned}
		\widehat{D} &= \widehat{G}(:,1:p), \quad
		\widehat{C} = \widehat{O}(1:m,:), \\
		\widehat{B} &= \widehat{Q}(:,1:p), \quad
		\widehat{A} = \widehat{O}^\dagger \widehat{H}^{+} \widehat{Q}^\dagger,
	\end{aligned}
\end{equation}
where $\widehat{G}(:,1:p)$ is the left-most $ m \times p $ block of $ \widehat{G}$, and $\widehat{O}(1:m,:)$ is the top-most $ m \times n $ block of $ \widehat{O}$, and $\widehat{Q}(:,1:p)$ is the left-most $ n \times p $ block of $ \widehat{Q}$.

The same procedure applied to the true Markov matrix $G$ recovers a true balanced state-space realization of the~\eqref{linear_system}.

\begin{remark}
	Note that $D$ is a submatrix of $G$. Consequently,
	\begin{equation}
	    \| D - \widehat{D} \| \leq \| G - \widehat{G}\|.
	\end{equation}
	Thus, in the following, it suffices to focus on the identification of the matrices $A$, $B$, and $C$.
\end{remark}

\section{Finite Sample Analysis of System Poles}
\label{sec:main results}

Oymak and Ozay \cite{oymak2021revisiting} and Zheng and Li~\cite{zheng2020non} establish high-probability upper bounds for the identification error of state-space realizations in the single-trajectory and multiple-trajectory settings, respectively. However, the corresponding error analysis for system poles, namely the eigenvalues of $A$, using the Ho-Kalman algorithm remains lacking. Therefore, in the following, we aim to address this gap.

\subsection{Single-trajectory setting}

In this subsection, we provide finite-sample high-probability upper bounds for pole estimation using the single-trajectory $\{(u_k,y_k)\mid 0 \leq k \leq T-1\}$, based on Oymak and Ozay~\cite{oymak2021revisiting}. Further, we analyze the sample complexity required to achieve a constant estimation error.

Before proceeding, we introduce some notation adopted from~\cite{oymak2021revisiting}. Define $\Phi(A)$ as the ratio between the exponents of the spectral norm and the square-root of the spectral radius of $A$:
\begin{equation}
    \Phi(A) \triangleq \sup_{\tau \geq 0} \frac{\|A^\tau\|}{\spr(A)^{\tau/2}}.
\end{equation}
We further define $\Gamma_\infty$ as the steady-state covariance matrix of the state sequence $x_k$, namely,
\begin{equation}
    \Gamma_\infty \triangleq \sum_{i=0}^\infty \sigma_w^2 A^i (A^\top)^i + \sigma_u^2 A^i B B^\top (A^\top)^i.
\end{equation}
Finally, we introduce $\sigma_e$, which is given by
\begin{equation}
    \sigma_e \triangleq \Phi(A) \| C A^{K-1}\| \sqrt{\frac{K \| \Gamma_\infty\|}{1-\spr(A)^K}}.
\end{equation}

\begin{theorem}[Theorem V.3 in~\cite{oymak2021revisiting}]\label{state-space realization of single sample trajectory}
	For the \eqref{linear_system}, under the conditions of Assumption \ref{observable and controllable} and \ref{system order is known}, suppose $\sigma_u = 1$ and $\sigma_v$, $\sigma_e$, $\sigma_w \|F\|$ are bounded by constants, and $\spr(A)^K \leq 0.99$. Let $H$ and $\widehat{H}$ be the Hankel matrices constructed from the Markov parameter matrices $G$ and $\widehat{G}$ with the single-trajectory $\{(u_k,y_k)\mid 0 \leq k \leq T-1\}$, respectively, via~\eqref{H}. Let $\overline{A}, \overline{B}, \overline{C}$ be the the state-space realization corresponding to the output of Ho–Kalman algorithm with input $G$ and $\widehat{A}, \widehat{B}, \widehat{C}$ be the state-space realization corresponding to output of Ho-Kalman algorithm with input $\widehat{G}$.
	Let $\overline{T}_0 = Kq\log^2(Kq)$, where $q = p+m +n$. Suppose\footnote{Since $ L $ is the best rank-$n$ approximation of $H^{-}$, this means that $\sigma_{n}(H^-) = \sigma_{n}(L)$, thus here we can use $ \sigma_{n}(H^-) $ instead of $\sigma_{n}(L)$ in~\cite{oymak2021revisiting}.} $\sigma_{n}(H^-) > 0 $ and perturbation obeys
	\begin{equation}
	    \|L-\widehat{L}\| \leq \frac{ \sigma_{n}(H^-)}{2},
	\end{equation}
	and sample size parameter $\overline{T}$ obeys
	\begin{equation}
	    \frac{\overline{T}}{\log^2 (\overline{T}q)} \sim \Omega \left( \frac{K\overline{T}_0}{\sigma_{n}^2(H^-)} \right).
	\end{equation}
	Then, with high probability (same as Theorem 3.1 in~\cite{oymak2021revisiting}), there exists a unitary matrix $\mathcal{U} \in \mathbb{R}^{n \times n}$, such that
    \begin{multline}
        \max \left\{\|\overline{C}-\widehat{C} \mathcal{U}\|_{\rm F}, \|O-\widehat{O} \mathcal{U}\|_{\rm F}, \|\overline{B}-\mathcal{U}^\top \widehat{B} \|_{\rm F}, \right. \\
		\left.	\| Q-\mathcal{U}^\top \widehat{Q} \|_{\rm F} \right\} \leq \frac{\sqrt{\mathcal{C}nK}\log(\overline{T}q)}{\sqrt{\sigma_{n}(H^-)}} \sqrt{\frac{\overline{T}_0}{\overline{T}}},
    \end{multline}
	where $\mathcal{C}$ is a constant.
	Furthermore, hidden state matrices $\widehat{A}, \overline{A}$ satisfy
	\begin{equation}
	    \|\overline{A}-\mathcal{U}^\top \widehat{A} \mathcal{U}\|_{\rm F} \leq \frac{\mathcal{C}\sqrt{nK}\log(\overline{T}q)\|H\|}{\sigma_{n}^2(H^-)}  \sqrt{\frac{\overline{T}_0}{\overline{T}}}.
	\end{equation}
\end{theorem}

\begin{remark}
    Oymak and Ozay~\cite{oymak2021revisiting} point out that, with high probability, the finite sample estimation errors of the state-space realization decay at a rate of at least $1/\sqrt{\overline{T}}$.
\end{remark}

To characterize the gap between the spectra of matrices $\overline{A}$ and $\widehat{A}$, we first introduce the Hausdorff distance~\cite{hausdorff1914grundzuge} as follows.

\begin{definition}[Hausdorff Distance~\cite{hausdorff1914grundzuge}]
	Given $\mathcal{A} = (\alpha_{ij}) \in \mathbb{C}^{n \times n}$ and $\mathcal{B} = (\beta_{ij}) \in \mathbb{C}^{n \times n}$, suppose that $\lambda(\mathcal{A}) =  \{\lambda_1(\mathcal{A}), \cdots, \lambda_n(\mathcal{A})\}$ and $\lambda(\mathcal{B}) = \{\mu_1(\mathcal{B}), \cdots, \mu_n(\mathcal{B})\}$ are the spectra of matrix $\mathcal{A}$ and $\mathcal{B}$ respectively, then 
	\begin{equation}
	    d_{\rm H}(\mathcal{A},\mathcal{B}) \triangleq \max\{{\rm sv}_{\mathcal{A}}(\mathcal{B}), {\rm sv}_{\mathcal{B}}(\mathcal{A})\}
	\end{equation}
	is defined as the Hausdorff distance between the spectra of matrix $\mathcal{A}$ and $\mathcal{B}$, where 
	\begin{equation}
	    {\rm sv}_{\mathcal{A}}(\mathcal{B}) \triangleq \max_{1\leq j\leq n} \min_{1\leq i\leq n} | \lambda_i(\mathcal{A})- \mu_j(\mathcal{B})|
	\end{equation}
	is the spectrum variation of $\mathcal{B}$ with respect to $\mathcal{A}$.
\end{definition}

\begin{remark}
	The geometric meaning of $s_{\mathcal{A}}(\mathcal{B})$ can be explained as follows. Let $\mathcal{D}_i \triangleq \{z\in \mathbb{C} \mid |z-\lambda_i(\mathcal{A}) | \leq \gamma  \}$, $i = 1,\cdots,n$, then $s_{\mathcal{A}}(\mathcal{B}) \leq \gamma$ means that $ \lambda(\mathcal{B}) \subseteq \bigcup_{i=1}^n \mathcal{D}_i$. On the other hand, it can be shown that the Hausdorff distance is a metric on $\{\lambda(\mathcal{A})\mid \mathcal{A} \in \mathbb{C}^{n\times n}\}$.
\end{remark}

The following theorem provides high-probability upper bounds for pole estimation using the single-trajectory $\{(u_k,y_k)\mid 0 \leq k \leq T-1\}$, based on the results of Theorem~\ref{state-space realization of single sample trajectory}.

\begin{theorem}\label{pole of single sample trajectory}
	For the \eqref{linear_system}, under the conditions of Assumption \ref{observable and controllable} and \ref{system order is known}, suppose $\sigma_u = 1$ and $\sigma_v$, $\sigma_e$, $\sigma_w \|F\|$ are bounded by constants, and $\spr(A)^K \leq 0.99$. Let $H$ and $\widehat{H}$ be the Hankel matrices constructed from the Markov parameter matrices $G$ and $\widehat{G}$ with the single-trajectory $\{(u_k,y_k)\mid 0 \leq k \leq T-1\}$, respectively, via~\eqref{H}. Let $\overline{A}, \overline{B}, \overline{C}$ be the the state-space realization corresponding to the output of Ho–Kalman algorithm with input $G$ and $\widehat{A}, \widehat{B}, \widehat{C}$ be the state-space realization corresponding to output of Ho-Kalman algorithm with input $\widehat{G}$.
	Let $\overline{T}_0 = Kq\log^2(Kq)$, where $q = p+m +n$. Suppose $\sigma_{n}(H^-) > 0 $ and perturbation obeys
	\begin{equation}
	    \|L-\widehat{L}\| \leq \frac{ \sigma_{n}(H^-)}{2},
	\end{equation}
	and sample size parameter $\overline{T}$ obeys
	\begin{equation}
	    \frac{\overline{T}}{\log^2 (\overline{T}q)} \sim \Omega \left( \frac{K\overline{T}_0}{\sigma_{n}^2(H^-)} \right).
	\end{equation}
	Then, with high probability (same as Theorem 3.1 in~\cite{oymak2021revisiting}), we have
	\begin{equation}
	    d_{\rm H}(\widehat{A}, \overline{A}) \leq 
		 (\Delta + 2\|\overline{A}\|)^{1-\frac{1}{n}} \Delta^{\frac{1}{n}},
	\end{equation}
	where $\mathcal{C}$ is a constant, and
	\begin{equation}
	    \Delta =  \frac{\mathcal{C}\sqrt{nK}\log(\overline{T}q)\|H\|}{\sigma_{n}^2(H^-)}  \sqrt{\frac{\overline{T}_0}{\overline{T}}}.
	\end{equation}
\end{theorem}

Theorem~\ref{pole of single sample trajectory} derives
finite-sample bounds for pole estimation using the single-trajectory $\{(u_k,y_k)\mid 0 \leq k \leq T-1\}$ and demonstrates that the estimation error decreases at a rate of at least \( \mathcal{O}(T^{-1/2n}) \), where \( T = K + \overline{T} - 1 \) denotes the length of the single sampling trajectory.

On the other hand, based on Theorem~\ref{state-space realization of single sample trajectory}, Oymak and Ozay~\cite{oymak2021revisiting} point out that, ignoring logarithmic factors, when \( \overline{T}_0 \sim \Theta (Kq) \), for state-space realizations, in order to achieve a constant estimation error with high probability, the sample size parameter \( \overline{T} \) needs to satisfy \(  \overline{T} \sim \Omega (K^2 qn) \),  where $q = p+m +n$. Following a similar approach, one can derive the sample size complexity required to achieve a constant estimation error with high probability for system poles. 

However, the derivation of~\cite{oymak2021revisiting} implicitly assumes the Hankel matrix $H^-$ is well-conditioned, neglecting the critical dependence on $\sigma_{n}(H^-)$, i.e., the minimum singular value of the Hankel matrix $H^-$. This omission is nontrivial, as $\sigma_{n}(H^-)$ fundamentally governs the numerical stability of realization algorithms, this will become evident in the following. To bridge this gap, the following lemma provides a theoretically rigorous characterization of $\sigma_{n}(H^-)$.

\begin{theorem}\label{ill-conditioned}
    For the \eqref{linear_system}, under the condition of Assumption \ref{observable and controllable}, if the matrix $A$ is stable (or marginally stable) and possesses distinct real eigenvalues, then the $n$-th largest singular value of $H^-$ satisfies
	\begin{equation}\label{L}
		\sigma_{n}(H^-) \leq 2\overline{\delta} nK\sqrt{pm} \rho^{-\max \left\{
			\frac{\left\lfloor \frac{n-1}{2m} \right\rfloor}{\log (2mK_1)}, \frac{\left\lfloor \frac{n-1}{2p} \right\rfloor}{\log (2pK_2)}	
			\right\} },
	\end{equation}
	where $ \rho \triangleq e^{\frac{\pi^2}{4}} \approx 11.79$, and $K=K_1+K_2 +1$, and 
	\begin{equation}
	    \overline{\delta} \triangleq \Big(\max_{i,j}|b_{ij}|\Big)\Big(\max_{i,j}|c_{ij}|\Big).
	\end{equation}
    is defined as the product of the largest absolute entry of matrix $B$ and that of matrix $C$.
\end{theorem}

\begin{remark}
    Theorem~\ref{ill-conditioned} reveals that under the given conditions above, $\sigma_{n}(H^-)$ decays \textbf{super-polynomially} with respect to $\max\{n/m, n/p\}$.
\end{remark}

For the \eqref{linear_system},  when matrix $A$ is stable or marginally stable and all of its eigenvalues are distinct and real, the results of Theorem~\ref{pole of single sample trajectory} and Theorem~\ref{ill-conditioned} reveal that, ignoring the logarithmic factors, if the sample parameter satisfies \( \overline{T}_0 \sim \Omega\left( Kq\right) \), where $q = p+m+n$, the required sample parameter \( \overline{T} \) to achieve a constant estimation error for the system poles with high probability must satisfy
\begin{equation}\label{appro_1}
	\overline{T}  \sim \Omega \left( \frac{q}{n pm} \varrho^{\max \left\{\frac{\left\lfloor \frac{n-1}{2m} \right\rfloor}{\log (2K_1m)}, \frac{\left\lfloor \frac{n-1}{2p} \right\rfloor}{\log (2K_2p)}\right\}}  \right),
\end{equation}
where \( \varrho \triangleq \rho^4 = e^{\pi^2} \approx 19333.69 \), and the bound is derived based on Lemma~\ref{upper bound of H and F}. Note that the length of the sample trajectory satisfies \( T = K + \overline{T} - 1 \), which indicates that, achieving a constant estimation error requires a \textbf{super-polynomial} sample size $T$ in $\max\{n/m, n/p\}$, where $n/m$ and $n/p$ denote the state-to-output and state-to-input dimension ratios, respectively.

\subsection{Multiple-trajectory setting}

In this subsection, similar to the previous one, we provide finite-sample high-probability upper bounds for pole estimation using multiple trajectories $\{(u_k^{(i)},y_k^{(i)})\mid 0 \leq k \leq K-1, 1 \leq i \leq N\}$, based on Zheng and Li~\cite{zheng2020non}. We further analyze the sample complexity required to achieve a constant estimation error.

\begin{theorem}[Theorem 1 in \cite{zheng2020non}]\label{Markov parameter of multiple sample trajectories}
	For the \eqref{linear_system}, under the conditions of Assumption \ref{observable and controllable} and \ref{system order is known}, fix any $0<\delta<1$. If the number of sample trajectories $N \geq 8 p K+4(p+n+m+ 4) \log (3 K / \delta)$, we have with probability at least $1-\delta$,
	\begin{equation}
	    \|\hat{G}-G\| \leq \frac{\sigma_v C_1+\sigma_w C_2}{\sigma_u} \sqrt{\frac{1}{N}},
	\end{equation}
	where
	\begin{align}\label{C_1}
		C_1 &= 8  \sqrt{2 K(K+1)(p+m) \log (27 K / \delta)}, \\ \label{C_2}
		C_2 &= 16\|F\| \sqrt{\left(\frac{K^3}{3}+\frac{K^2}{2}+\frac{K}{6}\right) 2(p+n) \log (27 K / \delta)}.
	\end{align}
\end{theorem}

\begin{corollary}[Corollary 1 in \cite{zheng2020non}]\label{state-space realization of multiple sample trajectories}
	For the \eqref{linear_system}, under the conditions of Assumption \ref{observable and controllable} and \ref{system order is known}, for $N$ and $K$ sufficiently large, there exists a unitary matrix $\mathcal{U}$, and a constant $C_3$ depending on system parameters $(\overline{A},\overline{B},\overline{C})$, the dimension $(n,p,m,K)$ and $\sigma_u, \sigma_w, \sigma_v$, and logarithmic factor of $N$ such that we have with high probability
	\begin{equation}
	    \max \left\{ \| \widehat{A} - \mathcal{U}\overline{A}\mathcal{U}^\top, \|\widehat{B} - \mathcal{U}\overline{B} \|, \|\widehat{C}- \overline{C} \mathcal{U}^\top\|  \right\} 
		\leq \frac{C_3}{\sqrt{N}},
	\end{equation}
	where $(\widehat{A},\widehat{B},\widehat{C})$ is the output of the Ho-Kalman algorithm on the least-squares estimation $\widehat{G}$ from~\eqref{G_multi}.
\end{corollary}


\begin{remark}
	It is worth noting that Corollary~\ref{state-space realization of multiple sample trajectories} in~\cite{zheng2020non} is derived by combining Theorem~\ref{Markov parameter of multiple sample trajectories} with the results from Section 5.3 of~\cite{sarkar2019near}. However, neither~\cite{zheng2020non} nor its extended version~\cite{zheng2020nonarxiv} provide an explicit expression for \( C_3 \). In fact, by leveraging Theorem~\ref{Markov parameter of multiple sample trajectories} along with Lemma~\ref{lemma_1} and Theorem~\ref{lemma_2} from~\cite{oymak2021revisiting}, one can similarly obtain results analogous to Corollary~\ref{state-space realization of multiple sample trajectories}, establishing that  the finite sample estimation errors of the state-space realization decay at a rate of at least $1/\sqrt{N}$. As demonstrated in the following corollary.
\end{remark}

\begin{corollary}\label{state-space realization of multiple sample trajectories_1}
	For the \eqref{linear_system}, under the conditions of Assumption \ref{observable and controllable} and \ref{system order is known}, fix any $0<\delta<1$. If the number of sample trajectories $N \geq 8 p K+4(p+n+m+ 4) \log (3 K/ \delta)$, we have with probability at least $1-\delta$,
	\begin{multline}
			\| \overline{A} -  \mathcal{U}^\top \widehat{A} \mathcal{U}  \|_{\rm F} \leq  \frac{9\sqrt{n}}{\sigma_{n}(H^-)} \left( \frac{ \|H^+\|}{\sigma_{n}(H^-)} + 2\right) \\ \cdot
			 \sqrt{\min\{K_1,K_2+1\}} \frac{\sigma_v C_1+\sigma_w C_2}{\sigma_u} \sqrt{\frac{1}{N}},
	\end{multline}
	where $C_1$ and $C_2$ are respectively defined in~\eqref{C_1} and~\eqref{C_2}, and $(\widehat{A},\widehat{B},\widehat{C})$ is the output of the Ho-Kalman algorithm on the least-squares estimation $\widehat{G}$ from~\eqref{G_multi}.
\end{corollary}


Further, the following theorem provides a high-probability upper bound for the identification error of the system poles using multiple trajectories $\{(u_k^{(i)},y_k^{(i)})\mid 0 \leq k \leq K-1, 1 \leq i \leq N\}$, based on the results of Corollary~\ref{state-space realization of multiple sample trajectories_1}.

\begin{theorem}\label{pole of multiple sample trajectories}
	For the \eqref{linear_system}, under the conditions of Assumption \ref{observable and controllable} and \ref{system order is known}, fix any $0<\delta<1$. If the number of sample trajectories $N \geq 8 p K+4(p+n+m+ 4) \log (3K / \delta)$, we have with probability at least $1-\delta$,
	\begin{equation}
	    d_{\rm H}(\widehat{A}, \overline{A}) \leq 
		(\Delta' + 2\|\overline{A}\|)^{1-\frac{1}{n}} (\Delta')^{\frac{1}{n}},
	\end{equation}
	where $(\widehat{A},\widehat{B},\widehat{C})$ is the output of the Ho-Kalman algorithm on the least-squares estimation $\widehat{G}$ from~\eqref{G_multi}, and
	\begin{multline}
		\Delta' =  \frac{9\sqrt{n}}{\sigma_{n}(H^-)} \left( \frac{ \|H^+\|}{\sigma_{n}(H^-)} + 2\right) \sqrt{\min\{K_1,K_2+1\}} \\
		\cdot \frac{\sigma_v C_1+\sigma_w C_2}{\sigma_u} \sqrt{\frac{1}{N}},
	\end{multline}
	and $C_1$ and $C_2$ are respectively defined in~\eqref{C_1} and~\eqref{C_2}.
\end{theorem}

Theorem~\ref{pole of multiple sample trajectories} derives
finite-sample bounds for pole estimation using multiple trajectories $\{(u_k^{(i)},y_k^{(i)})\mid 0 \leq k \leq K-1, 1 \leq i \leq N\}$, and demonstrates that the estimation error of the system poles decreases at a rate of at least \( \mathcal{O}(N^{-1/2n}) \), where $N$ denotes the number of sample trajectories. 

By combining the results of  Theorem~\ref{pole of multiple sample trajectories} and Theorem~\ref{ill-conditioned}, it follows that for the \eqref{linear_system},  when matrix $A$ is stable or marginally stable and all of its eigenvalues are distinct and real, ignoring the logarithmic factors, the required sample parameter \( N \) to achieve a constant estimation error for the system poles with high probability must satisfy
\begin{equation}\label{appro_2}
	N  \sim \Omega \left( \frac{nK^3}{p} \varrho^{\max \left\{\frac{\left\lfloor \frac{n-1}{2m} \right\rfloor}{\log (2K_1m)}, \frac{\left\lfloor \frac{n-1}{2p} \right\rfloor}{\log (2K_2p)}\right\}}  \right),
\end{equation}
where \( \varrho \triangleq \rho^4 = e^{\pi^2} \approx 19333.69 \), and the bound is derived based on Lemma~\ref{upper bound of H and F}. This implies that, achieving a constant estimation error requires a \textbf{super-polynomial} sample size $K \times N$ in $\max\{n/m, n/p\}$, where $n/m$ and $n/p$ denote the state-to-output and state-to-input dimension ratios, respectively.


\section{Numerical Results}
\label{sec:numerical results}

In this section, we first consider the identification of the classical two-mass spring-damper system \cite{liu2013fast} using the Ho-Kalman algorithm. The system consists of two point masses, $m_1 = 1$ kg and $m_2 = 1$ kg, interconnected in series via linear springs and viscous dampers. Each mass is also coupled to a fixed boundary via a spring-damper pair, as illustrated in \autoref{fig:enter-label}.

\begin{figure}[htbp]
    \centering
    \begin{tikzpicture}[
       spring/.style={  
        thick,  
        decoration={  
            coil,  
            aspect=0.5,  
            segment length=1mm,  
            amplitude=1mm,  
            pre length=3mm,  
            post length=3mm  
        },
        decorate  
    },
   wall/.style={  
        thick,
        pattern=north east lines,  
        minimum height=1.5cm,  
        minimum width=1pt  
    },
    wheel/.style={  
        circle,
        fill=white,
        draw=black,
        thick,
        minimum size=4pt,  
        outer sep=0pt
    },
    ground/.style={  
        thick,
        pattern=north east lines,  
        minimum height=0.1cm,  
        anchor=north  
    },
     damper/.style={
        draw, 
        thick,
        line width=1pt,
        minimum width=0.3cm,
        minimum height=0.2cm,
        inner sep=0pt,
        outer sep=0pt
    }
]

\def\wallsep{7cm}     
\def\masswidth{1.2cm} 
\def\massheight{0.8cm} 
\def\groundlevel{-1cm}  
\def\wheeloffset{0.3cm} 
\def\dampershift{0.3cm} 

\node (ground) [ground, minimum width=\wallsep, anchor=north] at (\wallsep/2, \groundlevel-5) {};
\draw[thick] (ground.north west) -- (ground.north east);

\node (left wall) [wall, anchor=east] at (0, \groundlevel+16.5) {};  
\draw[thick] (left wall.north east) -- (left wall.south east);  

\node (right wall) [wall, anchor=west] at (\wallsep, \groundlevel+16.5) {};  
\draw[thick] (right wall.north west) -- (right wall.south west);  

\node (mass1) [draw, thick, minimum width=\masswidth, minimum height=\massheight, anchor=south] at (2, \groundlevel) {$m_1$};
\fill (1.4,\massheight/2+\groundlevel+6) circle (1pt);  
\fill (2.6,\massheight/2+\groundlevel+6) circle (1pt);  
\fill (1.4,\massheight/2+\groundlevel-6.5) circle (1pt); 
\fill (2.6,\massheight/2+\groundlevel-6.5) circle (1pt); 

\fill (0,\massheight/2+\groundlevel+6) circle (1pt); 
\fill (0,\massheight/2+\groundlevel-6.5) circle (1pt); 
\fill (7,\massheight/2+\groundlevel+6) circle (1pt); 

\node [wheel] at ([xshift=-\wheeloffset]mass1.south) {};  
\node [wheel] at ([xshift=\wheeloffset]mass1.south) {};   

\node (mass2) [draw, thick, minimum width=\masswidth, minimum height=\massheight, anchor=south] at (5, \groundlevel) {$m_2$};
\fill (4.4,\massheight/2+\groundlevel+6) circle (1pt);  
\fill (5.6,\massheight/2+\groundlevel+6) circle (1pt);  
\fill (4.4,\massheight/2+\groundlevel-6.5) circle (1pt); 

\node [wheel] at ([xshift=-\wheeloffset]mass2.south) {};  
\node [wheel] at ([xshift=\wheeloffset]mass2.south) {};   

\draw[spring] (0,\massheight/2+\groundlevel+6) -- (1.4,\massheight/2+\groundlevel+6)
    node [midway, above=0.5mm] {\small $k_1$};

\draw[spring] (2.6,\massheight/2+\groundlevel+6) -- (4.4,\massheight/2+\groundlevel+6)
    node [midway, above=0.5mm] {\small $k_2$};

\draw[spring] (5.6,\massheight/2+\groundlevel+6) -- (7,\massheight/2+\groundlevel+6)
    node [midway, above=0.5mm] {\small $k_3$};

\coordinate (dampStart1) at ([yshift=-\dampershift]0,\massheight/2+\groundlevel+2);
\coordinate (dampEnd1) at ([yshift=-\dampershift]1.4,\massheight/2+\groundlevel+2);

\node [damper, anchor=center] (d1) at ($(dampStart1)!0.5!(dampEnd1)$) {};

\draw[thick] (dampStart1) -- (d1.west);
\draw[thick] (d1.east) -- (dampEnd1);

\draw[thick] (d1.north east) -- ++(0.1cm,0);
\draw[thick] (d1.south east) -- ++(0.1cm,0);

\node[below] at (1.1,-0.8) {\small $c_1$};

\coordinate (dampStart2) at ([yshift=-\dampershift]2.6,\massheight/2+\groundlevel+2);
\coordinate (dampEnd2) at ([yshift=-\dampershift]4.4,\massheight/2+\groundlevel+2);

\node [damper, anchor=center] (d2) at ($(dampStart2)!0.5!(dampEnd2)$) {};

\draw[thick] (dampStart2) -- (d2.west);
\draw[thick] (d2.east) -- (dampEnd2);

\draw[thick] (d2.north east) -- ++(0.1cm,0);

\draw[thick] (d2.south east) -- ++(0.1cm,0);

\node[below] at (1.1+2.8,-0.8) {\small $c_2$};

\draw[thick, ->,black] (2,-0.2+0.2) -- ++(0.8,0) 
    node[above] {$q_1$};
\draw[thick, ->, black] (2-0.8,-0.2+0.4) -- ++(0.8,0) 
    node[pos=0.2, above] {$f_1$};
\draw[thick, ->, black] (5,-0.2+0.2) -- ++(0.8,0) 
    node[above] {$q_2$};
\draw[thick, ->, black] (5-0.8,-0.2+0.4) -- ++(0.8,0) 
    node[pos=0.2, above] {$f_2$};

\draw[thick] (2,-0.2) -- (2,-0.2+0.6);

\draw[thick] (5,-0.2) -- (5,-0.2+0.6
);

\end{tikzpicture}
    \caption{A two-mass spring-damper system.}
    \label{fig:enter-label}
\end{figure}

\begin{figure*}[bthp]
    \centering
    \begin{subfigure}[t]{0.3\linewidth}
        \centering
        \includegraphics[width=\linewidth]{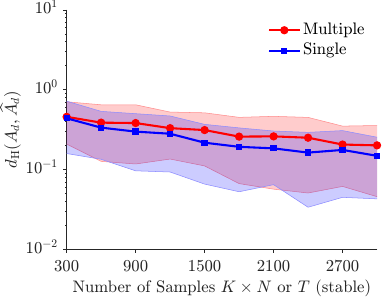}
        \caption{Stable system}
        \label{fig:stable}
    \end{subfigure}%
    \hfill
    \begin{subfigure}[t]{0.3\linewidth}
        \centering
        \includegraphics[width=\linewidth]{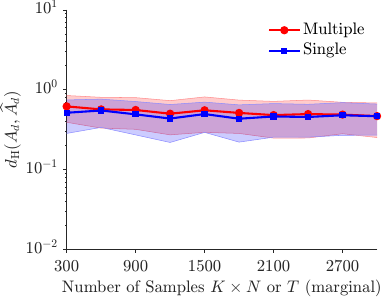}
        \caption{Marginally stable system}
        \label{fig:marginal}
    \end{subfigure}%
    \hfill
    \begin{subfigure}[t]{0.3\linewidth}
        \centering
        \includegraphics[width=\linewidth]{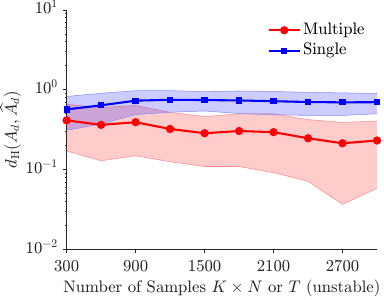}
        \caption{Unstable system}
        \label{fig:unstable}
    \end{subfigure}
    \hfill

    \begin{subfigure}[t]{0.25\linewidth}
        \centering
        \includegraphics[width=\linewidth]{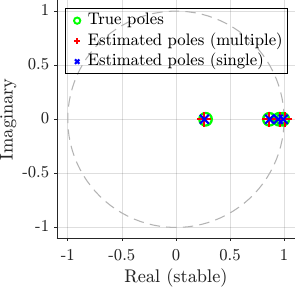}
        \caption{Stable system}
        \label{fig:stable_pole}
    \end{subfigure}%
    \hfill
    \begin{subfigure}[t]{0.25\linewidth}
        \centering
        \includegraphics[width=\linewidth]{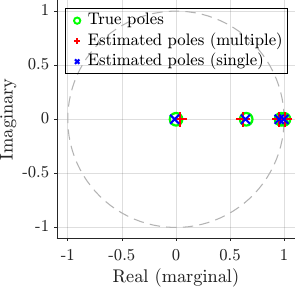}
        \caption{Marginally stable system}
        \label{fig:marginal_pole}
    \end{subfigure}%
    \hfill
    \begin{subfigure}[t]{0.25\linewidth}
        \centering
        \includegraphics[width=\linewidth]{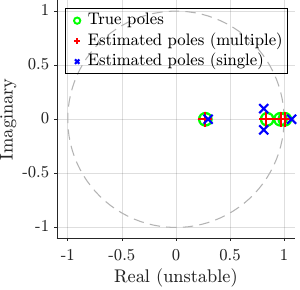}
        \caption{Unstable system}
        \label{fig:unstable_pole}
    \end{subfigure}
     
    \caption{Comparison of system pole estimation errors under different sample sizes: (a) stable system; (b) marginally stable system; (c) unstable system. In each subfigure, the red curve represents the Ho-Kalman algorithm from~\cite{zheng2020non}, which utilizes multiple trajectories (sample size $K \times N$); while the blue curve corresponds to the  Ho-Kalman algorithm from~\cite{oymak2021revisiting}, based on the single-trajectory (sample size $T$).
    Subfigures (d)–(f) depict, at the points where the estimation error attains its minimum, the distribution of the true poles (green `\(\circ\)') together with their estimates: red `\(+\)' for multiple trajectories and blue `\(\times\)' for the single-trajectory case, plotted in the complex plane.
    }
    \label{fig:all}
\end{figure*}

Let $q_1(t)$ and $q_2(t)$ denote the horizontal displacements of the masses $m_1$ and $m_2$, respectively. The control forces $f_1(t)$ and $f_2(t)$ are applied directly to the respective masses, while the system outputs are defined as their displacements. Applying Newton's second law yields the following equations of motion:
\begin{align}
    m_1 \ddot{q}_1 &= -k_1 q_1 + k_2(q_2 - q_1) - c_1 \dot{q}_1 + c_2(\dot{q}_2 - \dot{q}_1) + f_1(t), \\
    m_2 \ddot{q}_2 &= -k_2(q_2 - q_1) - k_3 q_2 - c_2(\dot{q}_2 - \dot{q}_1) + f_2(t).
\end{align}
Defining the state vector, control input vector and output vector as $x(t) = \begin{bmatrix}
        q_1(t) & \dot{q}_1(t) & q_2(t) & \dot{q}_2(t)
    \end{bmatrix}^\top$, $u(t) = \begin{bmatrix}
        f_1(t) & f_2(t)
    \end{bmatrix}^\top$, $y(t) = \begin{bmatrix}
        q_1(t)  & q_2(t) 
    \end{bmatrix}^\top$,
the system dynamics can be expressed in the continuous-time state-space form:
\begin{align}
    \dot{x}(t) &= A_cx(t) + B_cu(t), \\
    y(t) &= C_cx(t),
\end{align}
where 
$
A_c = \begin{bmatrix}
0 & 1 & 0 & 0 \\
-\frac{k_1 + k_2}{m_1} & -\frac{c_1 + c_2}{m_1} & \frac{k_2}{m_1} & \frac{c_2}{m_1} \\
0 & 0 & 0 & 1 \\
\frac{k_2}{m_2} & \frac{c_2}{m_2} & -\frac{k_2 + k_3}{m_2} & -\frac{c_2}{m_2}
\end{bmatrix}, \quad
B_c = \begin{bmatrix}
0 & \frac{1}{m_1} & 0 & 0\\
0 & 0 & 0 & \frac{1}{m_2}
\end{bmatrix}^\top, \quad
C_c = \begin{bmatrix}
1 & 0 & 0 & 0 \\
0 & 0 & 1 & 0
\end{bmatrix}
$.

To facilitate identification and digital implementation, the continuous-time system is converted to a discrete-time system using zero-order hold (ZOH) with a sampling period of $T_s = 0.1$ seconds. The resulting discrete-time state-space model is:
\begin{equation}\label{discrete-time system}
    \begin{aligned}
        x_{k+1} &= A_d x_k + B_d u_k + w_k, \\
        y_k &= C_d x_k + v_k,
    \end{aligned}
    \tag{\textsf{Test-system}}
\end{equation}
where $A_d = e^{A_c T_s}$, $B_d = \left( \int_0^{T_s} e^{A_c \tau} d\tau \right) B_c$, $C_d = C_c$, and $w_k \sim \mathcal{N}(\mathbf{0},10^{-4}\mathbf{I}_4)$, $v_k \sim \mathcal{N}(\mathbf{0},10^{-4}\mathbf{I}_2)$ are the i.i.d. process and measurement noise respectively.

For the~\eqref{discrete-time system}, we consider three different parameter sets corresponding to distinct stability categories—stable, marginally stable, and unstable.
\begin{itemize}
    \item Stable: $k_1 = 0.5$ N/m, $k_2 = 0.7$ N/m, $k_3 = 0.6$ N/m; $c_1 = c_2 = 5$ N$\cdot$s/m. Matrix $A_d$ has the spectrum $\lambda(A_d) = \{0.27,0.99,0.95,0.86\}$.
    \item Marginally stable: $k_1 = 0.5$ N/m, $k_2 = 0.7$ N/m, $k_3 = 0.6$ N/m; $c_1 = 60$ N$\cdot$s/m, $c_2 = 5$ N$\cdot$s/m. Matrix $A_d$ has the spectrum $\lambda(A_d) = \{0.001,0.65,0.97,1.00\}$.
    \item Unstable\footnote{This negative spring constant $k_1 = -0.7$ N/m emulates active control elements that generate displacement-dependent forces in phase with motion, inducing positive feedback.}: $k_1 = -0.7$ N/m, $k_2 = 0.7$ N/m, $k_3 = 0.6$ N/m; $c_1 = c_2 = 5$ N$\cdot$s/m. Matrix $A_d$ has the spectrum $\lambda(A_d) = \{0.27,0.84,0.97,1.01\}$.
\end{itemize}

In the experiments, we set the input (control force) $u_k \sim \mathcal{N}(\mathbf{0},\mathbf{I}_2)$. For the multiple-trajectory setting, each trajectory length is fixed at $K = 15$, while the number of trajectories, $N$, varies from $20$ to $200$, resulting in a sample size $K \times N$ ranging from $300$ to $3000$. In the single-trajectory case, the trajectory length (sample size), $T$, increases from $300$ to $3000$. The parameters \( K_1 \) and \( K_2 \) of the Hankel matrix are set to $8$ and $6$, respectively. Each scenario is evaluated over $100$ independent trials.  

We evaluate the estimation error of system poles using the Hausdorff distance metric and visualize the results using shade error bar, as shown in~\autoref{fig:all}. In each subfigure, the red curve represents the Ho-Kalman algorithm from~\cite{zheng2020non}, which utilizes multiple trajectories  (sample size $K \times N$); while the blue curve corresponds to the Ho-Kalman algorithm from~\cite{oymak2021revisiting}, which is based on a single trajectory  (sample size $T$).

As shown in~\autoref{fig:stable}, the estimation error of system poles in the stable system decreases as the sample size increases. 
In~\autoref{fig:marginal}, the estimation error decreases more slowly for the marginally stable system. This is because the slow mode with a pole of $0.001$ contributes weakly to the output and is easily masked by noise, making it difficult to identify accurately with limited data.
\autoref{fig:unstable} shows that for the unstable system, error decreases with multiple trajectories, but does not decrease with a single trajectory, as the single-trajectory Ho-Kalman algorithm is not suitable for unstable systems~\cite{oymak2021revisiting}. 

On the other hand, \autoref{fig:stable_pole},~\autoref{fig:marginal_pole} and~\autoref{fig:stable_pole} respectively illustrate, for the stable, marginally stable, and unstable systems, the distribution of the true poles (green `\(\circ\)') together with their estimates, plotted in the complex plane at the points where the estimation error attains its minimum. In particular, red `\(+\)' denotes the estimates obtained from multiple trajectories, whereas blue `\(\times\)' denotes those obtained from a single trajectory.




\section{Conclusion}
\label{sec:conclusions}
This paper presents an error analysis of system pole estimation for discrete-time MIMO LTI systems using the Ho-Kalman algorithm with finite input/output sample data. Building upon prior works, we derive high-probability error bounds for system pole estimation. Specifically, we prove that the estimation error for an $n$-dimensional system decays at a rate of at least $\mathcal{O}(T^{-1/2n})$ in the single-trajectory setting with trajectory length $T$, and at least $\mathcal{O}(N^{-1/2n})$ in the multiple-trajectory setting with $N$ independent trajectories. Additionally, we demonstrated that achieving a constant estimation error for system poles requires a
super-polynomial sample size in \( \max\{n/m, n/p\} \), where $n/m$ and $n/p$ denote the state-to-output and state-to-input dimension ratios, respectively. In the future, we will explore the results in the continuous-time setting.

\section*{Appendix}
\label{sec:appendix}

\subsection{Some lemmas}

\begin{theorem}[Theorem 1 in \cite{elsner1985optimal}]\label{perturbation of eigenvalues}
	Given $\mathcal{A} = (\alpha_{ij}) \in \mathbb{C}^{n \times n}$, $\mathcal{B} = (\beta_{ij}) \in \mathbb{C}^{n \times n}$, suppose that $\lambda(\mathcal{A}) =  \{\lambda_1(\mathcal{A}), \cdots, \lambda_n(\mathcal{A})\}$ and $\lambda(\mathcal{B}) = \{\mu_1(\mathcal{B}), \cdots, \mu_n(\mathcal{B})\}$ are the spectra of matrices $A$ and $B$ respectively, then
	\begin{equation}
		{\rm sv}_{\mathcal{A}}(\mathcal{B}) \leq  \left( \| \mathcal{A} \| + \|\mathcal{B}\| \right)^{1-\frac{1}{n}} \| \mathcal{A} - \mathcal{B} \|^{\frac{1}{n}}.
	\end{equation}
\end{theorem}

Based on the definition of the Hausdorff distance and Theorem~\ref{perturbation of eigenvalues}, the following corollary can be directly derived.

\begin{corollary}\label{hausdorff_distance}
	Given $\mathcal{A} = (\alpha_{ij}) \in \mathbb{C}^{n \times n}$, $\mathcal{B} = (\beta_{ij}) \in \mathbb{C}^{n \times n}$, suppose that $\lambda(\mathcal{A}) =  \{\lambda_1(\mathcal{A}), \cdots, \lambda_n(\mathcal{A})\}$ and $\lambda(\mathcal{B}) = \{\mu_1(\mathcal{B}), \cdots, \mu_n(\mathcal{B})\}$ are the spectra of matrix $A$ and $B$ respectively, then the Hausdorff distance between the spectra of matrix $\mathcal{A}$ and $\mathcal{B}$ satisfy
	\begin{equation}
		d_{\rm H}(\mathcal{A},\mathcal{B}) \leq  \left( \| \mathcal{A} \| + \|\mathcal{B}\| \right)^{1-\frac{1}{n}} \| \mathcal{A} - \mathcal{B} \|^{\frac{1}{n}}.
	\end{equation}
\end{corollary}

This corollary indicates that the Hausdorff distance between the spectra of matrix $\mathcal{A}$ and $\mathcal{B}$ can be controlled by the distance between matrix $\mathcal{A}$ and $\mathcal{B}$ in the sense of matrix norm.



\begin{lemma}[Lemma 4 in~\cite{sun2023finite}]\label{krylov matrix}
	Given a $ n \times mp $ ($ p \leq  n $) matrix $ X_{n,mp} $ satisfying 
	\begin{equation}
		X_{n,mp} = \begin{bmatrix}
			W_p &
			DW_p &
			\cdots &
			D^{m-1}W_p
		\end{bmatrix},
	\end{equation}	
	where $ W_p \in \mathbb{R}^{n \times p} $ and $ D $ is a normal matrix with real eigenvalues, then the smallest singular value of $ X_{n,mp} $ satisfies
	\begin{equation}\label{krylov matrix_result}
		\sigma_{\min}(X_{n,mp}) \leq 4 \rho_{}^{-\frac{\left\lfloor \frac{\min\{n,m(p-[p]_*)\}-1}{2p} \right\rfloor}{\log (2mp)}} \|X_{n,mp}\|, 
	\end{equation}
	where $ \rho \triangleq e^{\frac{\pi^2}{4}} $, and $ [p]_* = 0 $ if $ p $ is even or $ p = 1 $ and is $ 1 $ if $ p $ is an odd number greater than $ 1 $.
\end{lemma}

\begin{lemma}\label{upper bound of H and F}
If the system \eqref{linear_system} is stable (or marginally stable), then the observability matrix $O$ defined in~\eqref{O}, the controllability matrix $Q$ defined in~\eqref{Q}, the Hankel matrix $H$ defined in~\eqref{H}, and the matrix $F$ defined in~\eqref{F} satisfy
    \begin{equation}
        \begin{aligned}
        \| O \| &\leq  \overline{c} \sqrt{K_1mn}, \qquad
        \| Q \| \leq \overline{b} \sqrt{(K_2+1)pn}, \\
        \| H \| &\leq \frac{1}{2}\,\overline{\delta} n \sqrt{pm}\,K, \quad
        \| F \| \leq \overline{c} \sqrt{mnK},
        \end{aligned}
    \end{equation}
where $\overline{c} \triangleq \max_{i,j} |c_{ij}|$ and $\overline{b} \triangleq \max_{i,j} |b_{ij}|$ denote the maximum absolute entries of the matrices $C$ and $B$, respectively, and $\overline{\delta} \triangleq \left(\max_{i,j}|b_{ij}|\right)\left(\max_{i,j}|c_{ij}|\right)$.
\end{lemma}

\begin{proof}
We first derive an upper bound for $\| O \|$. Since the stability (or marginal stability) of $A$ implies $\| A^k \| \leq 1$ for all $k$, it follows that
\begin{equation}
\begin{aligned}
    \| O \|^2 &\leq \| O \|_{\mathrm{F}}^2 
    = \sum_{k = 0}^{K_1-1} \| C A^k \|_{\mathrm{F}}^2 
    \leq \sum_{k = 0}^{K_1-1} \| C \|_{\mathrm{F}}^2 \| A^k \|^2 \\
    &\leq \sum_{k = 0}^{K_1-1} \overline{c}^2 mn
    = \overline{c}^2 K_1 mn,
\end{aligned}
\end{equation}
which yields $\| O \| \leq \overline c \sqrt{K_1 mn}$.  
The estimate for $Q$ follows analogously, giving $\| Q \| \leq \overline b \sqrt{(K_2+1) pn}$, and the same reasoning applied to $F$ yields $\| F \| \leq \overline c \sqrt{mnK}$.

For $H$, by norm inequalities we obtain
\begin{equation}
\begin{aligned}
    \| H \|^2 &\leq \| H \|_{\mathrm{F}}^2 = \sum_{k_1 = 0}^{K_1 - 1} \sum_{k_2 = 0}^{K_2} \left\| C A^{k_1 + K_2} B \right\|_{\mathrm{F}}^2 \\
    &\leq \sum_{k_1 = 0}^{K_1 - 1} \sum_{k_2 = 0}^{K_2} \| C \|_{\mathrm{F}}^2 \| B \|_{\mathrm{F}}^2 \| A^{k_1 + K_2} \|^2 \\
    &\leq \sum_{k_1 = 0}^{K_1 - 1} \sum_{k_2 = 0}^{K_2} \overline{\delta}^2 n^2 pm
    \leq \frac{1}{4} \, \overline{\delta}^2 n^2 pm K^2,
\end{aligned}
\end{equation}
where the third inequality uses the assumption that \( A \) is stable or marginally stable, so \( \| A^k \| \leq 1 \) for all \( k \), and the final inequality follows from the arithmetic–geometric mean bound: $K_1(K_2 + 1) \leq \left( \frac{K_1 + K_2 + 1}{2} \right)^2 = \frac{K^2}{4}$.
Taking square roots on both sides yields the desired bound on \( \| H \| \).
\end{proof}

\begin{lemma}[Lemma V.1 in~\cite{oymak2021revisiting}]\label{lemma_1}
	The Hankel matrices \( H \), \( \widehat{H} \), and $ L$, $\widehat{L}$ defined in Subsection~\ref{subsec:ho-kalman}, satisfy the following perturbation conditions:
	\begin{multline}
		\max \left\{\| H^+ - \widehat{H}^+ \|, \| H^- - \widehat{H}^- \| \right\} \leq \| H - \widehat{H} \| \\
		\leq \sqrt{\min\{K_1,K_2+1\}} \| G - \widehat{G}\|,
	\end{multline}
	and
	\begin{multline}
			\| L - \widehat{L} \| \leq 2\| H^- - \widehat{H}^- \|  \leq 2 \sqrt{\min\{K_1,K_2\}} \| G - \widehat{G}\|.
	\end{multline}
\end{lemma}

\begin{theorem}[Theorem V.2 in~\cite{oymak2021revisiting}]\label{lemma_2}
	For the \eqref{linear_system}, under the conditions of Assumption \ref{observable and controllable} and \ref{system order is known}, let $H$ and $\widehat{H}$ be the Hankel matrices constructed from the Markov parameter matrices $G$ and $\widehat{G}$ with multiple trajectories $\{(u_k^{(i)},y_k^{(i)})\mid 0 \leq k \leq K-1, 1 \leq i \leq N\}$, respectively, via~\eqref{H}. Let $\overline{A}, \overline{B}, \overline{C},\overline{D}$ be the the state-space realization corresponding to the output of Ho–Kalman algorithm with input $G$ and $\widehat{A}, \widehat{B}, \widehat{C},\widehat{D}$ be the state-space realization corresponding to output of Ho-Kalman algorithm with input $\widehat{G}$. Suppose $\sigma_{n}(L) > 0$ and perturbation obeys
	\begin{equation}
		\| L - \widehat{L} \| \leq \frac{\sigma_{n}(L)}{2}.
	\end{equation}
	Then there exists a unitary matrix $\mathcal{U} \in \mathbb{R}^{n \times n}$ such that
	\begin{equation}
		\| \overline{A} -  \mathcal{U}^\top \widehat{A} \mathcal{U}  \|_{\rm F} \leq \frac{9\sqrt{n}}{\sigma_{n}(L)} \left( \frac{\|L - \widehat{L}\|}{\sigma_{n}(L)} \|H^+\| + \| H^+ - \widehat{H}^+\|\right).
	\end{equation}
\end{theorem}

\subsection{The proof of Theorem~\ref{pole of single sample trajectory}}

\begin{proof}
    According to Theorem~\ref{state-space realization of single sample trajectory}, there exists a unitary matrix $\mathcal{U} \in \mathbb{R}^{n \times n}$ such that with high probability (same as Theorem 3.1 in~\cite{oymak2021revisiting}), $	\| \mathcal{U}^\top \widehat{A}\mathcal{U} - \overline{A} \| 
	\leq \Delta$.
    On the other hand, according to the triangle inequality, we obtain
	\begin{equation}
	    \|\mathcal{U}^\top \widehat{A}\mathcal{U} \|  \leq \Delta + \| \overline{A}\|.
	\end{equation}
	Based on Corollary~\ref{hausdorff_distance}, it can be obtained that 
	\begin{multline}
			d_{\rm H}(\widehat{A},\overline{A}) = d_{\rm H}(\mathcal{U}^\top \widehat{A}\mathcal{U},  \overline{A})  \\
			\leq  \left( \|\mathcal{U}^\top \widehat{A}\mathcal{U}\| + \|\overline{A}\| \right)^{1-\frac{1}{n}} \| \mathcal{U}^\top \widehat{A}\mathcal{U} - \overline{A}\|^{\frac{1}{n}} \\
            \leq (\Delta + 2\|\overline{A}\|)^{1-\frac{1}{n}} \Delta^{\frac{1}{n}},
	\end{multline}
	where the first equality holds because the unitary matrix transformation does not change the spectrum of the matrix. 
\end{proof}

\subsection{The proof of Theorem~\ref{ill-conditioned}}
Theorem~\ref{ill-conditioned} can, in fact, be regarded as a corollary of Lemma 1 in \cite{sun2023finite}. Nevertheless, for the sake of completeness and self-contained exposition, we provide the detailed proof here.

\begin{proof}
    Note that $H^- = OQ^-$, where $O$ is defined as~\eqref{O}, and $Q^- \triangleq \begin{bmatrix}
		B & AB& \cdots & A^{K_2-1}B
	\end{bmatrix}$.
    Since the pair $(A,C)$ is observable and $(A,B)$ is controllable, and because $K_1, K_2 \geq n$, the matrix $O$ has full column rank $n$, while $Q^-$ has full row rank $n$.

    We first bound $\sigma_{n}(H^-) $ via $O$:
    \begin{multline}
        \sigma_{n}(H^-)  \leq \sigma_{n}(O) \| Q^- \| \leq 4\rho^{-\frac{\left\lfloor \frac{n-1}{2m} \right\rfloor}{\log (2mK_1)} }\|O\| \cdot\| Q^- \| \\
        \leq 
        2\overline{\delta} nK\sqrt{pm} \rho^{
			\frac{\left\lfloor \frac{n-1}{2m} \right\rfloor}{\log (2mK_1)} },
    \end{multline}
    where the second inequality follows from Lemma~\ref{krylov matrix}, and the final one is due to Lemma~\ref{upper bound of H and F}.

    Similarly, bounding $\sigma_{n}(H^-)$ via $Q^-$ yields
    \begin{multline}
        \sigma_{n}(H^-)  \leq \sigma_{n}(Q^-) \| O \| \leq 4\rho^{-\frac{\left\lfloor \frac{n-1}{2m} \right\rfloor}{\log (2pK_2)} }\|Q^-\| \cdot\| O \| \\
        \leq 
        2\overline{\delta} nK\sqrt{pm} \rho^{
			\frac{\left\lfloor \frac{n-1}{2m} \right\rfloor}{\log (2pK_2)} }.
    \end{multline}
\end{proof}

\subsection{The proof of Corollary~\ref{state-space realization of multiple sample trajectories_1}}

\begin{proof}
	By combining the results of Theorem~\ref{Markov parameter of multiple sample trajectories} with Lemma~\ref{lemma_1} and Theorem~\ref{lemma_2}, the proof can be completed.
\end{proof}

\subsection{The proof of Theorem~\ref{pole of multiple sample trajectories}}

\begin{proof}
    According to Corollary~\ref{state-space realization of multiple sample trajectories_1}, if the number of sample trajectories $N \geq 8 p K+4(p+n+m+ 4) \log (3 K/ \delta)$, there exists a unitary matrix $\mathcal{U} \in \mathbb{R}^{n \times n}$ such that with probability at least $1-\delta$, $	\| \mathcal{U}^\top \widehat{A}\mathcal{U} - \overline{A} \| 
	\leq \Delta'$.
    On the other hand, according to the triangle inequality, we obtain
	\begin{equation}
	    \|\mathcal{U}^\top \widehat{A}\mathcal{U} \|  \leq \Delta' + \| \overline{A}\|.
	\end{equation}
	Based on Corollary~\ref{hausdorff_distance}, it can be obtained that 
	\begin{multline}
			d_{\rm H}(\widehat{A},\overline{A}) = d_{\rm H}(\mathcal{U}^\top \widehat{A}\mathcal{U},  \overline{A})  \\
			\leq  \left( \|\mathcal{U}^\top \widehat{A}\mathcal{U}\| + \|\overline{A}\| \right)^{1-\frac{1}{n}} \| \mathcal{U}^\top \widehat{A}\mathcal{U} - \overline{A}\|^{\frac{1}{n}} \\
            \leq (\Delta' + 2\|\overline{A}\|)^{1-\frac{1}{n}} \Delta'^{\frac{1}{n}},
	\end{multline}
	where the first equality holds because the unitary matrix transformation does not change the spectrum of the matrix. 
\end{proof}


\bibliographystyle{elsarticle-num}
\bibliography{ref}

\balance  

\end{document}